\documentclass[aps,prl,twocolumn,10pt,nofootinbib]{revtex4-2}

\usepackage[T1]{fontenc}
\usepackage[utf8]{inputenc}
\usepackage{amsmath,amssymb,amsfonts,amsthm,bm,mathtools}
\usepackage{braket}
\usepackage{dsfont}
\usepackage{graphicx}
\usepackage{comment}
\usepackage{booktabs}
\newtheorem{theorem}{Theorem}
\usepackage{makecell}
\usepackage{bbm}
\usepackage{physics}
\usepackage{titletoc}
\usepackage[colorlinks,allcolors=blue]{hyperref}

\begin{document}

\title{How fast can a quantum gate be? Exact speed limits from geometry}

\author{Hunter Nelson}
\author{Edwin Barnes}
\affiliation{Department of Physics, Virginia Tech, Blacksburg, VA 24061, USA}
\affiliation{Virginia Tech Center for Quantum Information Science and Engineering, Blacksburg, VA 24061, USA}

\begin{abstract}
The speed of quantum evolution is limited under finite energy resources. While most quantum speed limits (QSLs) are formulated in terms of quantum states, they can be extended to the evolution operator itself, and thus impose fundamental limits on how quickly logical gate operations can be implemented on a quantum computer. Here, we derive a general, tight QSL that holds for any unitary evolution under the constraint that the spectral width of the Hamiltonian is bounded. We apply this result to obtain QSLs for several standard quantum gates, including Hadamard, CNOT, and Toffoli gates, finding that the QSL can vary significantly across different gates, including ones with the same entangling power. These findings can be understood geometrically using the Space Curve Quantum Control formalism, which maps unitary evolution to space curves in Euclidean space. In this formalism, the problem of finding QSLs is recast as the problem of finding minimal-length curves obeying a curvature bound. We find that time-optimal gates map to helices of varying dimensions, and that QSLs can be understood from the perspective of a bottleneck principle in which the operator that evolves the slowest governs the minimal gate time. 
\end{abstract}

\maketitle

\textbf{\textit{Introduction}}.
It has been known since the 1940s that the speed of quantum evolution is bounded when constraints on the Hamiltonian are imposed~\cite{deffner2017}. Most quantum speed limits (QSLs) that have been discovered to date are \emph{state}-based \cite{aharonov1961time,lloyd2000ultimate,pfeifer1993fast,giovannetti2003quantum,campaioli2018tightening,naseri2024quantum}, bounding how quickly a chosen initial state can evolve into an orthogonal or distinguishable state. Canonical examples include the Mandelstam--Tamm bound, governed by energy variance~\cite{mandelstam1945uncertainty}, and the Margolus--Levitin bound, governed by average energy above the ground state~\cite{margolus1998maximum}. These bounds were later unified~\cite{levitin2009fundamental}. 

QSLs for the time evolution operator itself were subsequently discovered. Existing formulations typically combine (i) a metric or distance on $SU(n)$ with (ii) an upper bound on an induced speed functional, often derived from a norm constraint on the Hamiltonian \cite{deffner2017,farmanian2024quantum,impens2025approaching,poggi2019geometric,russell2017geometry}. Such bounds are broadly applicable, but often not tight in practice. Moreover, they rarely identify \emph{which} dynamical obstruction is responsible for the minimal time required to implement a given target unitary. These issues are especially important for quantum information technologies, where finite coherence times penalize slow logic gates, while longer evolution times increase control errors, accumulated noise, and calibration drift. These limitations motivate the search for tight QSLs for unitaries, as well as a deeper understanding of their origin.    

In this Letter, we derive an \emph{exact and saturable} QSL for unitary evolution subject to a bound on the spectral width of the Hamiltonian. We obtain a closed-form minimal time $T_{\star,G}$ for realizing any gate $G{\in} SU(n)$ in the fully controllable setting in which the spectral-width bound is the only restriction imposed on the Hamiltonian. This spectral-width constraint is a natural choice given that wave function amplitudes evolve according to energy differences.  We obtain this QSL using the Space Curve Quantum Control (SCQC) formalism \cite{barnes2022dynamically,buterakos2021geometrical,zeng2018fastest}, which maps unitary evolution to space curves in Euclidean space. In SCQC, the gate time equals the curve length, while the spectral-width constraint translates to a bound on the curve’s curvature, recasting the search for QSLs into the problem of finding the shortest curves in Euclidean space that satisfy certain boundary conditions while obeying the curvature bound. We obtain QSLs for many of the most commonly used logic gates in quantum computing, such as the Hadamard, CNOT, and Toffoli gates, for which the minimal times are $T_{\star,\mathrm{Hadamard}}{=}T_{\star,\mathrm{CNOT}}{=}T_{\star,\mathrm{Toffoli}}{=}\pi/\Omega_\mathrm{max}$, where $\Omega_\mathrm{max}$ is the spectral-width bound. When viewed as space curves, the speed-limited evolutions that produce these gates correspond to helices of varying dimensions. We find that QSLs are governed by a \emph{bottleneck principle} in which the minimal time is set by the operator(s) that evolve the slowest. Our results provide the minimal possible time in which a target unitary can be realized and serve as a baseline for QSLs that arise when further restrictions are imposed on the Hamiltonian.

\textbf{\textit{Geometry of optimal evolutions}}.
We consider the time-optimal synthesis of a target gate $G {\in} SU(n)$ under a physically motivated energy resource constraint that we discuss below. Unitary evolution obeys $i\dot U(t){=}H(t)U(t)$, $U(0){=}I$, where $I$ is the identity, and $H(t)$ is fully controllable aside from the single resource constraint. General approaches to time-optimal gate synthesis include the quantum brachistochrone formulation \cite{wang2015quantum, carlini2005quantum,koike2022quantum}, the Pontryagin maximum-principle and Riemannian geometry formulations \cite{khaneja2001time,khaneja2002sub,garon2013time,van2017robust,boscain2021introduction}, and related geometric constructions \cite{nielsen2006optimal,nielsen2005geometric,nielsen2006quantum,gu2008quantum,dowling2006geometry}, which search for geodesics on $SU(n)$ defined by a speed functional of the Hamiltonian.

\begin{figure}[t]
\includegraphics[scale=.26]{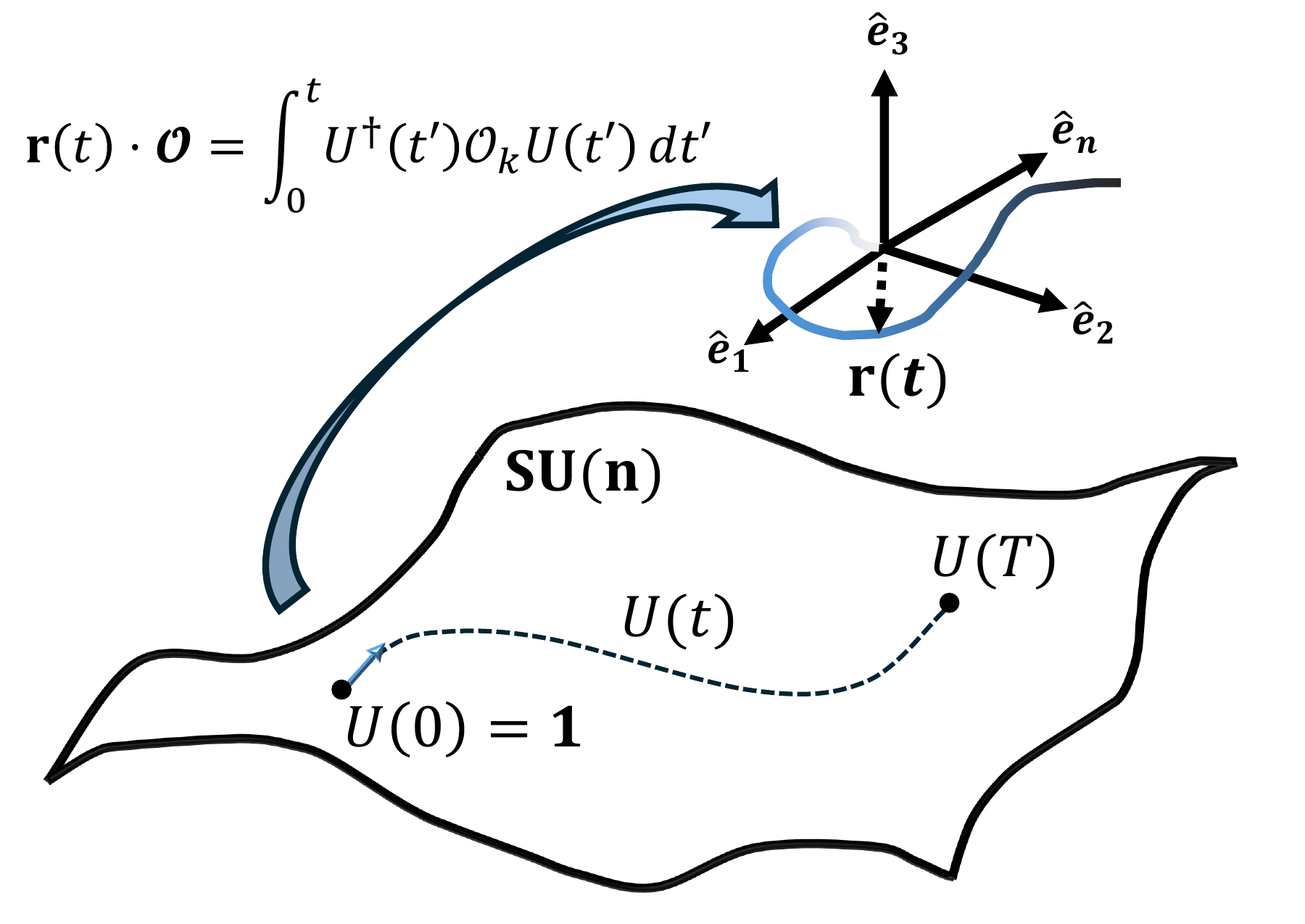}
\caption{Paths on the manifold $\mathbf{SU}(n)$ emanating from the identity and terminating at three different gates $U_1$, $U_2$, and $U_3$. A space curve in SCQC represents these paths in a Euclidean space $\mathbb R^{n^2-1}$, whose dimension equals the number of group generators.} \label{fig:manifold}
\end{figure}

We adopt a dramatically different approach based on the SCQC formalism, which maps unitary dynamics to curves in the Euclidean operator space $\mathfrak{su}(n) {\simeq} \mathbb R^{n^2-1}$ which is equipped with the Hilbert--Schmidt inner product $\|A\|{=}\sqrt{\mathrm{Tr}(A^\dagger A)}$. We further fix an orthonormal basis $\bm{\mathcal O}{=}(\mathcal O_1,\ldots,\mathcal O_{n^2{-}1})$ of $\mathfrak{su}(n)$ which allows us to write any traceless Hermitian operator
$A {\in}\mathfrak{su}(n)$ as a coordinate vector $\mathbf a{\in}\mathbb R^{n^2{-}1}$ via
$A{=}\bm{\mathcal O} {\cdot} \mathbf a$. Then for a unit-norm observable $\mathcal O$, we define a space curve $ \mathbf r_{\mathcal O}(t)$ as
\begin{equation}
\bm{\mathcal O}\cdot \mathbf r_{\mathcal O}(t)=\int_0^t dt'\,U^\dagger(t')\,\mathcal O\,U(t').
\label{eq:curve_def}
\end{equation}
Different choices of $\mathcal O$ generate different curves. The curve's arc length equals the total evolution time $T$, as can be seen by first differentiating both sides of Eq.~(\ref{eq:curve_def}) and observing that $\|\bm{\mathcal O} {\cdot} \dot{\mathbf r}_{\mathcal O}(t)\|{=}1$, $\forall t$. As a consequence, the curve has unit speed ($ds{=}\|\dot{\mathbf r}_{\mathcal O}\|dt{=}dt$), and its arc length equals the evolution time: $L{=}\int_0^T\|\bm{\mathcal O}{\cdot}\dot{\mathbf r}_{\mathcal O}(t)\|\,dt{=}T$.
An important geometric property of a curve is its curvature, which is the rotation rate of the tangent curve $\dot{\mathbf{r}}_{\mathcal O}(t)$: $\kappa_{\mathcal O}(t){=}\|\ddot{\mathbf{r}}_\mathcal{O}\|{=}\|\bm{\mathcal{O}} {\cdot} \ddot{\mathbf{r}}_{\mathcal O}\|
{=}\|[H(t),\mathcal O]\|$. Further details about SCQC can be found in Sec.~\ref{Subsupp:1.1 SCQC and Frenet-Serret equations}. 

The problem of finding the QSL for a gate $G$ is thus equivalent to the task of finding the shortest space curves whose final tangent points satisfy $\bm{\mathcal{O}}{\cdot}\dot{\mathbf{r}}_{\mathcal O}(T){=}U^\dagger(T)\,\mathcal O\,U(T){=}G^\dagger \mathcal O\,G$. However, the problem is not meaningful unless we impose a physical restriction on how strongly we drive the system. Options include bounding the norm of the Hamiltonian $||H(t)||$, or the individual drives that appear within it, but such bounds are generally not tight (see Sec.~\ref{Subsupp:1.3 Why the spectral width of H provides tighter QSLs that |H|}). Here, we instead choose to bound the instantaneous spectral-width $w(H(t)){:=}E_{\max}(H(t)){-}E_{\min}(H(t)){\le} \Omega_{\max}$, where $E_{\max}(H(t))$ and $E_{\min}(H(t))$ are the largest and smallest instantaneous eigenvalues of $H(t)$. We choose this resource constraint for three reasons: (i) It bounds the maximal time rate of change of wave function coefficients, or equivalently, evolution operator components (see Sec.~\ref{Subsupp:1.3 Why the spectral width of H provides tighter QSLs that |H|}). (ii) It is invariant to nonphysical overall energy shifts, $H(t){\to} H(t){+}\lambda(t)I$. (iii) It directly relates to the curvature formula via
$\|[H(t),\mathcal{O}]\|{\le} w(H(t))$, yielding the curvature bound $\kappa_{\mathcal O}(t){\le} w(H(t)){\le} \Omega_{\max}$.
Therefore, in SCQC this resource constraint becomes a \emph{turning-rate bound} on $\dot{\mathbf r}_\mathcal{O}(t)$, i.e., the tangent curve lives on the unit sphere and cannot rotate faster than $\Omega_{\max}$. 

As a demonstration of this bound, consider an operator $\mathcal{O}$ and target gate $G$ such that $\dot{\mathbf{r}}_\mathcal{O}(0){=}\mathbf a$ and $\dot{\mathbf r}_\mathcal{O}(T){=}\mathbf b$. The shortest curve that interpolates between these boundary conditions is a great-circle arc in the two-dimensional (2d) $\mathbf{a}$--$\mathbf{b}$ plane that saturates the spectral-width bound at all times: $\|\ddot{\mathbf r}(t)\|{=}\Omega_{\max}$. The length of this arc is $T_{2\text{d}}
{=}\arccos(\mathbf{a} \cdot \mathbf{b})/\Omega_{\max}$, and thus $T {\ge} T_{2\text{d}}$, 
which is tight provided $\mathcal{O}$ can indeed generate such an arc. For single-qubit gates this is always possible as any $U{\in} SU(2)$ can be written in the standard form:
$U{=}e^{-i(\theta/2)\,\hat{\mathbf n}{\cdot}\boldsymbol{\sigma}}$, whereby under $w(H){\le}\Omega_{\max}$ one may choose
$H{=}(\Omega_{\max}/2)\,\hat{\mathbf n}{\cdot}\boldsymbol{\sigma}$.
Any operator $\mathcal{O}{=}\hat{\mathbf{n}}'{\cdot}\boldsymbol{\sigma}$ with $\hat{\mathbf{n}}{\perp}\hat{\mathbf{n}}'$  undergoes a uniform planar rotation with curvature $\Omega_{\max}$, achieving $T{=}T_{\text{2d}}$. For example, if $\mathcal{O}{=}Z$ (so that $\mathbf{a}{=}\hat{z}$) and $G{=}X$, then $\mathbf{b}{=}{-}\hat{z}$, and the QSL is $T_{\star,X}{=}T_\mathrm{2d}{=}\pi/\Omega_\mathrm{max}$. For a Hadamard gate, the rotation axis is $\hat{\mathbf{n}}{=}(\hat{x}{+}\hat{z})/\sqrt{2}$. Choosing $\mathcal{O}{=}(X{-}Z)/\sqrt{2}$ so that $\hat{\mathbf{n}}'{\cdot}\hat{\mathbf{n}}{=}0$, we have $\mathbf{b}{=}{-}\mathbf{a}{=}{-}\hat{\mathbf{n}}'$, and therefore $T_{\star,\mathrm{Hadamard}}{=}T_{2 \mathrm d}{=}\pi/\Omega_{\max}$.

However, general multi-qubit gates do not always saturate this planar bound. This is because the SCQC tangent is constrained to remain on the adjoint orbit of $\mathcal O$.
For example, if $\mathcal O{=}ZZ$ the following constraint holds: $(U^\dagger ZZ U)^2{=}II$, $\forall t$, so the tangent cannot follow an arbitrary great-circle path in operator space. To make this concrete, 
consider CNOT, which maps the Pauli string $ZZ$ to $IZ$.
If the tangent could rotate in the $ZZ{-}IZ$ plane, i.e., $U^\dagger(t)\,ZZ\,U(t){=}\alpha(t)\,ZZ{+}\beta(t)\,IZ$, then the minimal time would be $T_{2\text{d}}{=}\pi/(2\Omega_{\max})$.  However, such an evolution is impossible since
$(\alpha ZZ+\beta IZ)^2{=}(\alpha^2+\beta^2)II+2\alpha\beta\,ZI\neq II$
along any continuous interpolation from $(\alpha,\beta){=}(1,0)$ to $(0,1)$.
Hence implementing $ZZ\mapsto IZ$ forces $\dot{\mathbf{r}}_{ZZ}(t)$ out of the plane, and similarly for $\mathbf{r}_{ZZ}(t)$. We can interpret this failure as resulting from the fact that the great circles connecting the identity to CNOT do not reside within Pauli planes in the space of 2-qubit operators. This begs the question of whether there exist other choices of $\mathcal{O}$ for which $\mathbf{r}_\mathcal{O}(t)$ does lie within a plane.

\textbf{\textit{Certifying sets, bottleneck operators, and a general quantum speed limit}}.
In order to answer this question, we first make the key observation that the space curve $\mathbf{r}_{\mathcal O}(t)$ associated with any single operator $\mathcal{O}$ does not uniquely determine $U(t)$ in general, as can be seen from the fact that any other evolution operator $U'(t){=}V(t)U(t)$, where $[V(t),\mathcal{O}]{=}0$, generates the same curve. Therefore, we need to consider a set of space curves $\{\mathbf{r}_{\mathcal O}(t)|\mathcal O{\in}\mathcal{G}\}$ generated by a set of operators $\mathcal G{\subset}\mathfrak{su}(n)$ to uniquely specify a particular $U(t)$ up to a global phase. We refer to any such set $\mathcal{G}$ as a \emph{certifying set} of operators for $U(t)$. For instance, $\mathcal G$ may be a full operator basis of $\mathfrak{su}(n)$, or a set of Pauli strings for multi-qubit systems. In Sec.~\ref{Subsupp:1.5 Two operators suffice to determine a unitary up to a global phase}, we prove that there always exist certifying sets containing only two operators, $|\mathcal{G}|{=}2$, and that these are the smallest possible sets. Any evolution operator that passes through $G$ must satisfy the simultaneous tangent endpoint constraints $\bm{\mathcal{O}}{\cdot}\dot{\mathbf{r}}_{\mathcal O}(T){=}U^\dagger(T)\,\mathcal O\,U(T){=}G^\dagger \mathcal O\,G{=}\text{Ad}_G(\mathcal O)$, $\forall\,\mathcal O{\in}\mathcal G$. Since all constraints must hold at the same final time, the gate duration obeys
\begin{align}
T\ \ge\ T_{\star,G}=\max_{\mathcal O\in\mathcal G}\ T_{\min}\!\big(\mathcal O {\to} \text{Ad}_G(\mathcal O)\big),
\label{eq:bottleneck_bound}
\end{align}
where $T_{\min}(\mathcal O {\to} \text{Ad}_G(\mathcal O))$ is the minimum time needed to evolve $\mathcal{O}$ to $\text{Ad}_G(\mathcal O)$. Thus, the QSL is determined by the slowest operator in the certifying set; we refer to this as the \emph{bottleneck principle}.  

We now derive an explicit formula for the QSL of any gate $G$ by computing the bottleneck bound, Eq.~\eqref{eq:bottleneck_bound}, for a specific certifying set: $\mathcal{G}{=}\{\mathcal{P}_{jk}\}{\cup}\{\mathcal{Q}_{jk}\}$, where $\mathcal{P}_{jk}{=}\tfrac{1}{\sqrt{2}}(\dyad{j}{k}{+}\dyad{k}{j})$ and $\mathcal{Q}_{jk}{=}{-}\tfrac{i}{\sqrt{2}}(\dyad{j}{k}{-}\dyad{k}{j})$, where $\ket{k}$ is an eigenstate of $G$ with eigenvalue $e^{i\phi_k}$, and $j{\ne} k$. In Sec.~\ref{Subsupp:1.6 Certifying set for deriving the quantum speed limit}, we prove that this is a certifying set for any evolution operator $U(t)$. (In fact, $\{\mathcal{P}_{jk}\}$ is by itself a certifying set for $n{\ge}3$.) We have $G^\dagger\mathcal{P}_{jk}G{=}\cos(\phi_j{-}\phi_k)\mathcal{P}_{jk}{+}\sin(\phi_j{-}\phi_k)\mathcal{Q}_{jk}$, and similarly for $G^\dagger\mathcal{Q}_{jk}G$. The shortest curve that has the boundary conditions $\dot{\mathbf{r}}_{\mathcal{P}_{jk}}(0){=}\mathcal{P}_{jk}$ and $\dot{\mathbf{r}}_{\mathcal{P}_{jk}}(T){=}G^\dagger\mathcal{P}_{jk}G$ is a great-circle arc in the plane spanned by $\mathcal{P}_{jk}$ and $\mathcal{Q}_{jk}$, and it is evident that the tangent must rotate by a net angle of $|\phi_j{-}\phi_k|$ along this curve. If we assume for the moment that there exists a Hamiltonian that can generate this arc, and that the spectral-width bound is saturated throughout this evolution, $\|\ddot{\mathbf{r}}_{\mathcal{P}_{jk}}(t)\|{=}\Omega_\mathrm{max}$, then it follows that $T_\mathrm{min}(\mathcal{P}_{jk}{\to}\mathrm{Ad}_G(\mathcal{P}_{jk})){=}|\phi_j{-}\phi_k|/\Omega_\mathrm{max}$. From Eq.~\eqref{eq:bottleneck_bound}, the QSL time $T_{\star,G}$ is obtained by maximizing this result over all values of $j$ and $k$:
\begin{align}
T\ge T_{\star,G}:=\frac{\Delta\phi_{\star,G}}{\Omega_{\max}} \label{eq:Tstar},
\end{align}
where $\Delta\phi_{\star,G}{=}\phi_\mathrm{max}{-}\phi_\mathrm{min}$, where $\phi_\mathrm{max}$ and $\phi_\mathrm{min}$ are the largest and smallest of the $\phi_k$, respectively. Note that, because the $\phi_k$ are unique only up to shifts by multiples of $2\pi$, we should always choose these shifts such that $\Delta\phi_{\star,G}$ is minimized.
Equation ~\eqref{eq:Tstar} is the central result of this paper. It is saturated by a constant Hamiltonian $H_{\star,G}
{:=}{-}\sum_k\tfrac{\phi_k}{T_{\star,G}} \dyad{k}{k}$, so that $e^{-iT_{\star,G} H_{\star,G}}{=}G$ up to a global phase. Defining $U(t){=}e^{-it H_{\star,G}}$, we have $U^\dagger(t)\mathcal{P}_{jk}U(t)=\cos((\phi_j{-}\phi_k)t/T_{\star,G})\mathcal{P}_{jk}+\sin((\phi_j{-}\phi_k)t/T_{\star,G})\mathcal{Q}_{jk}$, validating the assumption that there always exists a Hamiltonian that generates great-circle arc evolutions.

\begin{figure}[t] \includegraphics[scale=.3]{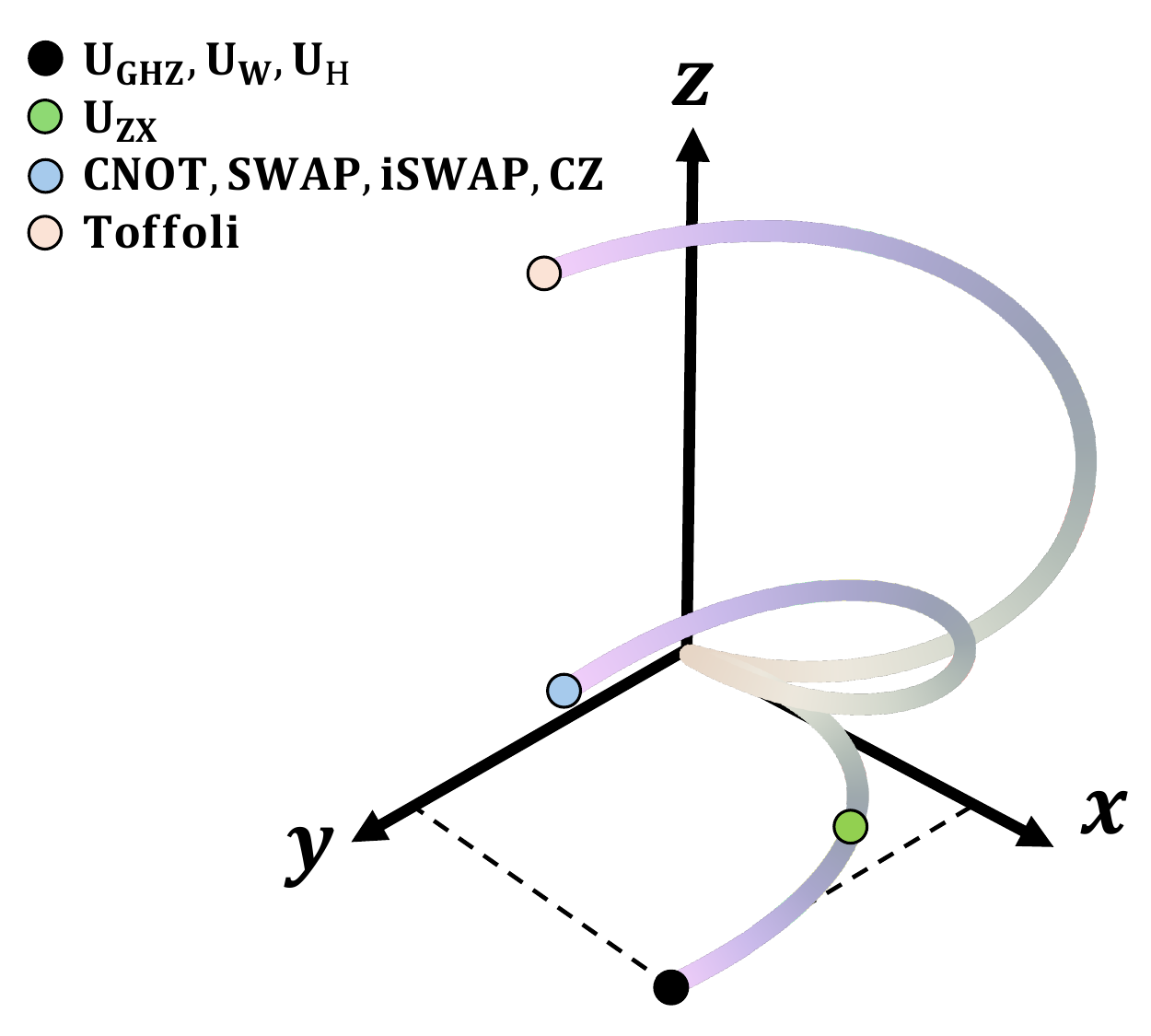} \caption{Time-optimal space curves generated by Pauli certifiers and corresponding to several commonly used gates in quantum information science. Extra curve dimensions represent longer QSL times for gates. For instance the curve realizing CNOT is a circular helix which necessarily dips into a third dimension. This leads to a longer QSL over the locally equivalent  $\text{U}_{\text{ZX}}$ gate, which admits a planar realization.} \label{fig:helices} \end{figure}

\textbf{\textit{Geometry of QSLs in the Pauli basis}}. 
The proof of the general QSL in Eq.~\eqref{eq:Tstar} shows that all time-optimal quantum evolutions that are constrained only by the spectral-width bound can be viewed geometrically as planar arcs. However, this is only true for certain certifying sets---we saw earlier that for the CNOT gate, certifying sets based on Pauli strings do not correspond to planar curves. To better understand these non-planar cases, we can use the time-optimal Hamiltonian $H_{\star,G}$.

Returning to the CNOT example, a time-optimal realization consistent with Eq.~\eqref{eq:Tstar} is given by: $U(t){=}e^{-itH_{\star,\mathrm{CNOT}}}$ with $H_{\star,\mathrm{CNOT}}{=}\frac{\Omega}{2}(ZI{+}IX{-}ZX)$,
from which we obtain $\mathrm{CNOT}{=}U(\pi/2\Omega)$ up to a global phase.
Here the tangent formed from $ZZ$ can be written in terms of a constant offset plus a uniform rotation in a single 2d plane: $U^\dagger(t)\,ZZ\,U(t)
=\frac{1}{2}(ZZ{+}IZ)
+\frac{1}{2}[\cos(2\Omega t)(ZZ{-}IZ)+\sin(2\Omega t)(ZY{-}IY)]$,
showing that the space curve is a 3d helix rather than a planar arc. The spectral width is $w(H_{\star,\mathrm{CNOT}}){=}2\Omega$, which exactly equals the tangent's rotation rate in the active plane. For a CNOT, the tangent must complete a $\pi$ rotation in the active plane,
$2\Omega T{=}\pi$, and saturating the width bound, $w(H_{\star,\mathrm{CNOT}}){=}2\Omega{=}\Omega_{\max}$, gives the minimal time: $T_{\star,\mathrm{CNOT}}{=}\pi/\Omega_{\max}
{>}\pi/(2\Omega_{\max}){=}T_{2\mathrm d}$. Note that this lower bound is above the minimal time of $\sqrt{3}/\Omega_\mathrm{max}$ reported in Ref.~\cite{farmanian2024quantum} for CNOT and, thus, tighter. Even though the tangent turning rate saturates the spectral-width bound, the fact that the curve is forced into a third dimension leads to a longer gate time. Also notice that if we remove one of the single-qubit terms from the Hamiltonian and instead consider $H_{\star}{=}\tfrac{\Omega}{2}(IX{-}ZX)$, which generates a gate locally equivalent to CNOT, then this time there exist Pauli strings that generate planar curves, e.g., $U^\dagger(t)YZU(t){=}\cos(\Omega t)YZ{+}\sin(\Omega t)YY$. However, the observable $ZZ$ still maps to a 3d helix, and the minimal time is still $T_\star{=}\pi/\Omega_\mathrm{max}$. Hence, even when some certifying observables admit planar motion, the overall gate time is controlled by the slowest observable as in Eq.~\eqref{eq:bottleneck_bound}.

Although CNOT does not saturate the planar bound, there are locally equivalent gates that do. For instance, the unitary $\text{U}_{ZX}(t){=}e^{-i\Omega tZX/2}$ is locally equivalent to CNOT at $t{=}\pi/(2\Omega)$, yet this time $\text{U}_{ZX}^\dagger(t)ZZ \text{U}_{ZX}(t){=}\cos(\Omega t)ZZ{+}\sin(\Omega t)IY$ is planar, and the minimal time equals the planar bound: $T_{\star,\text{U}_{ZX}}{=}T_{2\mathrm{d}}{=}\pi/(2\Omega_\mathrm{max})$.
Thus, the spectral-width-based QSL is not invariant under local gates in general.

\begin{table}[t]
\centering
\small
\setlength{\tabcolsep}{10.5pt}
\renewcommand{\arraystretch}{1.05}
\begin{tabular}{@{}lccc@{}}
\toprule
Gate $G$ & $\Delta\phi_{\star,G}$ & $T_{\star,G}$ & Geometry \\
\midrule
$\text{U}_{ZX}$ & $\pi/2$ & $\pi/(2\Omega_{\max})$ & circular arc\\
\makecell[l]{$\text{U}_{\text{H}}$, $\text{U}_{\mathrm{GHZ}}$,\\ $\text{U}_{\mathrm{W}}$} & $\pi$ & $\pi/\Omega_{\max}$ & circular arc\\
\makecell[l]{CNOT, CZ,\\ SWAP, iSWAP} & $\pi$ & $\pi/\Omega_{\max}$ & 3d helix\\
Toffoli & $\pi$ & $\pi/\Omega_{\max}$ & 3d helix\\
$\text{U}_{4\text{d}}$ & $3\pi/2$ & $3\pi/(2\Omega_{\max})$ & 4d helix\\
\bottomrule
\end{tabular}
\caption{Minimal gate times $T_{\star,G}$ under the spectral-width bound $w(H){\le}\Omega_{\max}$ and corresponding space curve geometries under Pauli certifiers.}
\label{tab:qsl_gates}
\end{table}

Based on these observations, we can organize gates into classes based on their QSL time $T_\star$ and the SCQC geometry induced by a width-saturating time-optimal generator $H_{\star}$ under a chosen certifying set. Because conjugation by a constant basis-change unitary acts as an orthogonal transformation in operator space, it rigidly rotates the space curve without changing its shape or length. Therefore, for classification purposes it is sufficient to work in the eigenstate basis of $H_{\star}$.

We begin with the commonly used two-qubit Clifford gates $\mathrm{CZ}$, $\mathrm{SWAP}$, $\mathrm{iSWAP}$, and $\mathrm{CNOT}$ under Pauli certifiers. These gates share the same QSL:
$T_\star{=}\pi/\Omega_{\max}$, while their bottleneck Pauli certifiers correspond to 3d helices. 
However, other multi-qubit gates with three-qubit entangling capabilities such as
$\text{U}_{\mathrm{GHZ}}{=}\exp\![i\frac{\pi}{2\sqrt{2}}(XXX{+}ZII)]$ and
$\text{U}_{\mathrm{W}}{=}\exp\![i\frac{\pi}{2\sqrt{3}}(XYZ{+}YZX{+}ZXY)]$
have the same spectral width as the Hadamard gate,
$\text{U}_{\text{H}}{=}\exp\![i\frac{\pi}{2\sqrt{2}}(X{+}Z)]$, namely $\Delta\phi_\star{=}\pi$, and hence
$T_\star=\pi/\Omega_{\max}$. In particular, for $\text{U}_{\text{H}}$, $\text{U}_{\text{GHZ}}$, and $\text{U}_{\text{W}}$, the diagonal frame of $H_{\star}$ consists of a single diagonal Pauli string, so any Pauli certifier that anticommutes with $H_{\star}$ produces a circular-arc space curve. The Toffoli gate achieves the same $T_\star$, but its bottleneck Pauli certifiers generate a distinct 3d helix (see Fig.~\ref{fig:helices}).
Finally, we consider the gate $\text{U}_{4\mathrm d}{=}\exp\!\big[i\frac{\pi}{4}(2ZI+IZ)\big]$, which has Pauli certifiers that evolve as a 4d helix, leading to the largest $T_\star$ among our examples: 
$T_{\star}{=}3\pi/(2\Omega_{\max})$. Representative space curves for all these gates are depicted in Fig.~\ref{fig:helices}.

Other certifier choices lead to different geometric descriptions. In particular, choosing certifiers $\mathcal{P}_{jk}$, $\mathcal{Q}_{jk}$ built from the energy eigenbasis of $H_{\star,G}$ yields purely planar circular-arc motion, as we saw above. This representation is depicted in Fig.~\ref{fig:circ_arcs}, where gates that require the same tangent rotation have longer QSLs due to smaller curvatures or, equivalently, a larger radius. Further discussion of these examples can be found in Sec.~\ref{Subsupp:2.2 Two qubit constructions}, and Sec.~\ref{Subsupp:2.3 Three qubit constructions} of the SM.

\textbf{\textit{Additional Hamiltonian constraints}}.
The general bound $T{\ge}T_{\star,G}{=}\Delta\phi_{\star,G}/\Omega_{\max}$ holds for any Hamiltonian subject to the spectral-width constraint. Most physical systems have additional structure that further constrains the Hamiltonian. SCQC maps such constraints into geometric restrictions on space curves, thereby providing a sharp and constructive way to quantify and diagnose the increase in minimal gate time beyond $T_\star$ due to these additional constraints.

\begin{figure}[t]
\includegraphics[scale=.3]{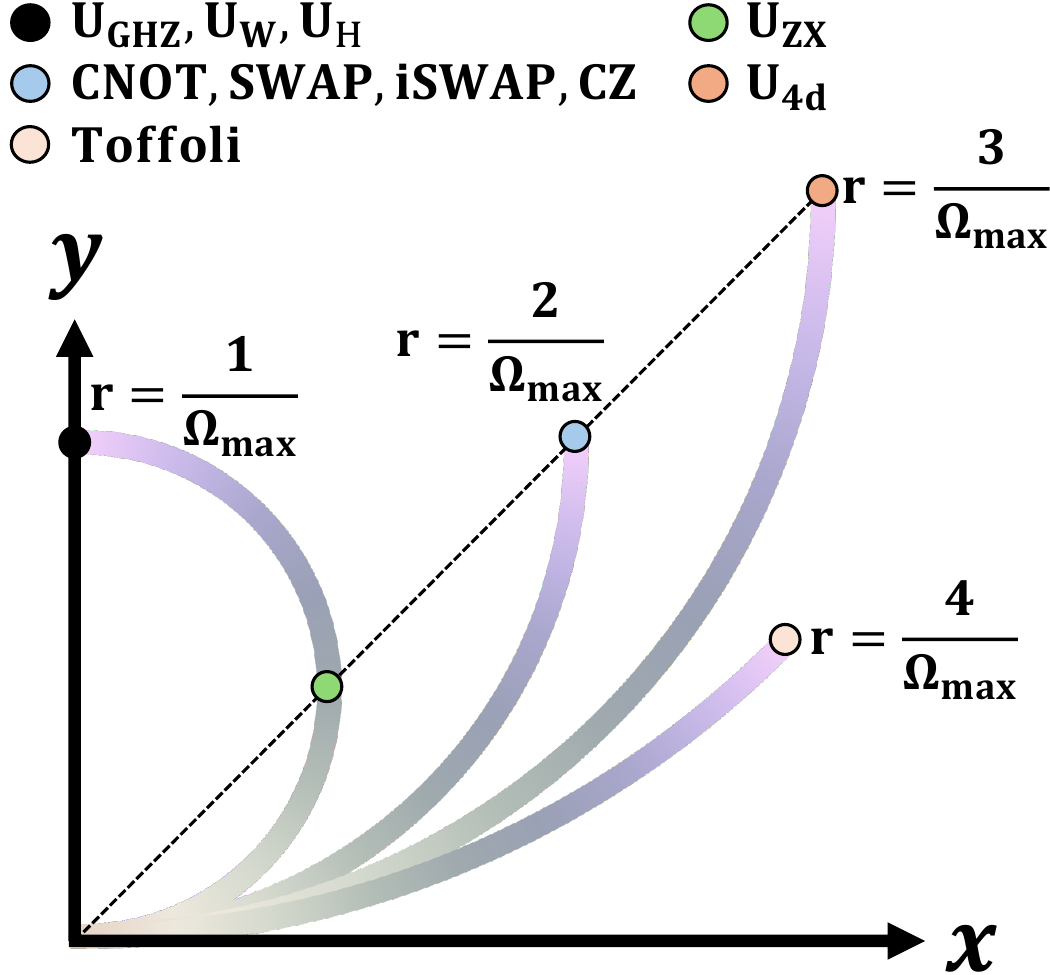}
\caption{Circular arc space curves generated by non-Pauli certifiers constructed from the eigenbasis of the optimal Hamiltonian. Here, for $\text{U}_{\text{ZX}}$, CNOT, and $\text{U}_{\text{4d}}$ a larger QSL geometrically corresponds to a circular arc with a larger radius for the same depicted tangent-endpoint angle of $\pi/4$.}
\label{fig:circ_arcs}
\end{figure}

To make this concrete, let $\mathcal{H} {\subset} \mathfrak{su}(n)$ denote the set of admissible Hamiltonians $H(t)$, and define the physically achievable minimal gate time as
$T_{{\mathcal H},G}{:=}\min\{T{:}\exists\,H(t){\in}\mathcal H,\ U(T){=}G\}$.
We then define the overhead factor:
$\eta_{{\mathcal H},G}{:=}T_{{\mathcal H},G}/T_{\star,G}{\ge} 1$.
While computing $T_{{\mathcal H},G}$ exactly is generally hard, every unitary that realizes the gate must satisfy the simultaneous constraints
$U^\dagger(T)\mathcal O\,U(T){=}G^\dagger\mathcal O\,G$ for all observables in a certifying set $\mathcal G$. Consequently:
\begin{equation}
T_{{\mathcal H},G}\ \ge\ \max_{\mathcal O\in\mathcal G}\ T_{\mathcal H,\min}\!\big(\mathcal O \to \text{Ad}_G(\mathcal O)\big),
\label{eq:arch_bottleneck}
\end{equation}
where $T_{\mathcal H,\min}(\mathcal O{\to}\text{Ad}_G(\mathcal O))$ is the minimal time required for a single observable to reach its target value
under $H(t){\in}\mathcal H$ and $w(H(t)){\le}\Omega_{\max}$. Thus, the bottleneck principle continues to govern minimal gate times in the case of additional Hamiltonian constraints.

This suggests a constructive diagnostic procedure. Given a target gate $G$ and an allowed set of Hamiltonians $\mathcal H$, first compute 
$T_{\star,G}{=}\Delta\phi_{\star,G}/\Omega_{\max}$. Next, choose a certifying set $\mathcal G$ and identify a bottleneck observable $\mathcal O_{\rm bottleneck}$.
Then check whether there exists an allowed Hamiltonian $H(t){\in}{\mathcal H}$ such that the tangent curve associated with $\mathcal O_{\rm bottleneck}$ remains in a single 2d plane.
If not, then it must be that $\eta_{{\mathcal H},G}{>}1$.
This procedure isolates the origin of the increased gate time, thereby providing insights into how to improve gate speeds.

\textbf{\textit{Conclusion.}}
We presented a general, tight lower bound on the time needed to implement a target unitary for Hamiltonians that are constrained only by their spectral width. We employed Space Curve Quantum Control to show that quantum evolutions that saturate this bound correspond to minimal-length space curves in Euclidean space that obey an upper bound on their curvature. We computed this minimal time for several of the most commonly used logical gate operations, including Hadamard, CNOT, SWAP, and Toffoli, finding that gates acting on different numbers of qubits can share the same speed limit, while gates with the same entangling power can have speed limits that differ. The corresponding space curves inhabit two or more dimensions, depending on the target unitary and on which set of observables are used to define them. In the case of Pauli observables, we found that target unitaries subject to the same quantum speed limit can correspond to space curves of differing dimensionality. Our results shed light on the ultimate speed with which quantum operations can be implemented, and they serve as a starting point for calculating minimal gate times when additional constraints on the Hamiltonian are present or when additional goals, such as noise cancellation, are desired.

\section*{Acknowledgments}

This work is supported by the Office of Naval Research, grant no. N00014-25-1-2125.

\bibliography{refs}

\clearpage
\onecolumngrid

\begin{center}
{\LARGE\bfseries Supplemental Material\par}
\vspace{0.5em}
\end{center}

\setcounter{secnumdepth}{2}
\setcounter{tocdepth}{2}

\setcounter{section}{0}
\setcounter{subsection}{0}
\setcounter{equation}{0}
\setcounter{figure}{0}
\setcounter{table}{0}

\renewcommand{\thesection}{S\arabic{section}}
\renewcommand{\thesubsection}{S\arabic{section}.\arabic{subsection}}
\renewcommand{\theequation}{S\arabic{equation}}
\renewcommand{\thefigure}{S\arabic{figure}}
\renewcommand{\thetable}{S\arabic{table}}

\makeatletter
\renewcommand{\p@subsection}{}
\renewcommand{\p@subsubsection}{}
\makeatother

\makeatletter
\renewcommand{\@seccntformat}[1]{\csname the#1\endcsname\quad}
\makeatother

\contentsmargin{2.2em}

\titlecontents{section}
  [0pt]
  {\addvspace{0.4em}}
  {\makebox[2.8em][l]{\thecontentslabel}}
  {}
  {\titlerule*[0.5pc]{.}\contentspage}

\titlecontents{subsection}
  [2em]
  {}
  {\makebox[3.6em][l]{\thecontentslabel}}
  {}
  {\titlerule*[0.5pc]{.}\contentspage}

\startcontents[supplement]
\printcontents[supplement]{}{1}{\setcounter{tocdepth}{2}}

\section{Space curve quantum control and quantum speed limits}
\label{Supp:1 SCQC and QSLs}
The space curve quantum control framework leverages the correspondence between quantum evolution via unitary dynamics and geometric space curves in Euclidean space \cite{barnes2022dynamically, buterakos2021geometrical, zeng2018fastest}. Historically, this framework has been used to study dynamically corrected gates. However, the framework is useful in other contexts as well and here we focus on its use in the study of quantum speed limits and gate classification. In this framework space curve lengths equal the evolution time, while resource bounds relate to geometric invariants of the space curve such as curvature. Finally, the curve is uniquely specified by a Frenet-Serret frame, which is a set of orthonormal vectors that evolve along the space curve. The boundary conditions of this frame fix the gate that is obtained at the final time. 

\subsection{SCQC and the Frenet-Serret equations}
\label{Subsupp:1.1 SCQC and Frenet-Serret equations}

Fix an orthonormal basis $\bm{\mathcal O}=(\mathcal O_1,\ldots,\mathcal O_{n^2-1})$ of $\mathfrak{su}(n)$ with $\frac{1}{n}\text{Tr}(\mathcal O_i\mathcal O_j)=\delta_{ij}$.
For any observable $\mathcal O_k$ with $\|\mathcal O_k\|=1$, define the SCQC curve:
\begin{align}
\bm{\mathcal O}\cdot \mathbf r_k(t)=\int_0^t dt'\,U^\dagger(t')\,\mathcal O_k\,U(t'),
\qquad
\bm{\mathcal O}\cdot \dot{\mathbf r}_k(t)=U^\dagger(t)\mathcal O_k U(t).
\label{eq:SI_SCQC_def}
\end{align}
Unitary conjugation preserves $\|\cdot\|$, hence $\|\dot{\mathbf r}_k(t)\|=1$ and $\mathrm{length}[\mathbf r_k]=T$. The space curve itself is uniquely determined by its generalized curvatures via the Frenet-Serret equations. These equations track the dynamics of an orthonormal frame of basis vectors that evolve along the space curve. This construction provides a basis-independent way to quantify how many operator-space dimensions are activated by a given certifying set of observables under the optimal generator of a chosen gate. This in turn dictates what quantum speed limit class a gate belongs to.

The Frenet-Serret frame for a constant generator can be constructed as follows. Start with the orthonormal basis element $\mathcal{O}_k$ associated with the space curve $\mathbf{r}(t)$, and define: $\mathcal F_1 \equiv \mathcal O_k$. The Frenet-Serret frame can then be constructed through the following recursion:
\begin{align}
\widetilde{\mathcal F}_{j+1}&= i[H,\mathcal F_j]+\kappa_{j-1}\mathcal F_{j-1},
\qquad \kappa_0\equiv 0,\qquad \mathcal F_0\equiv 0,\\
\kappa_j&=\|\widetilde{\mathcal F}_{j+1}\|,
\qquad
\mathcal F_{j+1}=\widetilde{\mathcal F}_{j+1}/\kappa_j,
\label{eq:SI_Lanczos}
\end{align}
until termination ($\kappa_\ell=0$). This process is equivalent to the Gram-Schmidt process but for the set of linearly dependent operators:
\begin{align}
\mathcal K(\mathcal O;H)
:=\mathrm{span}\Big\{\mathcal O,\ (i\,\mathrm{ad}_{H})\mathcal O,\ (i\,\mathrm{ad}_{H})^2\mathcal O,\ldots\Big\}
\end{align}
Then $\{\mathcal F_1,\ldots,\mathcal F_\ell\}$ is an orthonormal moving frame on the adjoint orbit of
$\mathcal O$, and the adjoint dynamics closes as:
\begin{align}
\frac{d}{dt}
\begin{pmatrix}
U^\dagger\mathcal F_1U\\
U^\dagger\mathcal F_2U\\
\vdots\\
U^\dagger\mathcal F_\ell U
\end{pmatrix}
=
\begin{pmatrix}
0&\kappa_1& & & \\
-\kappa_1&0&-\kappa_2& & \\
&\kappa_2&0&\ddots& \\
& &\ddots&\ddots&-\kappa_{\ell-1}\\
& & &\kappa_{\ell-1}&0
\end{pmatrix}
\begin{pmatrix}
U^\dagger\mathcal F_1U\\
U^\dagger\mathcal F_2U\\
\vdots\\
U^\dagger\mathcal F_\ell U
\end{pmatrix}.
\label{eq:SI_FS_general}
\end{align}
Here, the tangent curve is given by $\bm{\mathcal{O}}\cdot \dot{\mathbf{r}} = U^\dagger\mathcal F_1U$ and $\kappa_j$ are the generalized curvatures of the space curve. For constant $H$, the $\{\kappa_j\}$ are constant. Here the dimensionality of the geometry depends on both the Hamiltonian and observable corresponding to the tangent curve:
\begin{equation}
\ell(\mathcal O;H):=\dim\mathcal K(\mathcal O;H).
\end{equation}
For any nontrivial constant Hamiltonian $H$ there always exists a choice of $\mathcal{O}$ with $\ell=2$. For a constant $H$, and tangent curve $\bm{\mathcal O} \cdot \dot{\bm{r}}=U^\dagger(t)\mathcal O U(t)$ the dynamics can be written: $\bm{\mathcal O} \cdot\ddot{\bm{r}}(t)=A\bm{\mathcal O} \cdot\dot{\bm{r}}(t)$ with:
\begin{align}
A_{jl}=\frac{i}{n} \text{Tr}\left( [H,\mathcal O_j]\mathcal O_l \right),
\label{eq:SI_Adef}
\end{align}
where $A$ is constant, real, and antisymmetric. A well known mathematical result is that for any such $A$ there always exists $\Lambda\in SO(n^2-1)$ such that:
\begin{align}
A=\Lambda\left(\bigoplus_{p=1}^m
\begin{pmatrix}
0 & \kappa_p\\
-\kappa_p & 0
\end{pmatrix}
\oplus 0\right)\Lambda^T,
\label{eq:SI_Ablocks}
\end{align}
where in the adapted basis, each active $2\times2$ block describes a circular arc with curvature $\kappa_p$ \cite{horn2012matrix}.
Therefore we can always re-describe space curves with constant generalized curvatures, i.e. higher dimensional helices defined by Eq. (\ref{eq:SI_Lanczos}), as a direct sum of orthogonal circular arcs. 

The observables that yield such circular arc tangent curves can be constructed by starting with the Hamiltonian in the energy eigenbasis: $H=\sum_a E_a\ket{a}\bra{a}$ with $\sum_a E_a=0$, and defining the standard $\mathfrak{su}(n)$ basis:
\begin{align}
X_{ab}=\frac{1}{\sqrt{2}}\big(\ket{a}\bra{b}+\ket{b}\bra{a}\big),\qquad
Y_{ab}=\frac{-i}{\sqrt{2}}\big(\ket{a}\bra{b}-\ket{b}\bra{a}\big),\qquad (a<b).
\end{align}
Then
\begin{align}
[H,X_{ab}]=i(E_a-E_b)Y_{ab},\qquad
[H,Y_{ab}]=-i(E_a-E_b)X_{ab},
\label{eq:SI_XYblocks}
\end{align}
so each pair $(X_{ab},Y_{ab})$ forms a rotating $2$-plane with rate $\kappa_{p}=|E_a-E_b|$. This tells us that the spectral gaps are directly related to the curvatures of these circular arcs. Therefore:
\begin{align}
\max_{p}\kappa_{p}=\max_{p}|E_a-E_b|=w(H),
\label{eq:SI_gap_width}
\end{align}
and under the Hamiltonian width bound, all planar curvatures satisfy $\kappa_{p}\le \Omega_{\max}$.

\subsection{A commutator bound implied by spectral width}
\label{Subsupp:1.2 A commutator bound implied by spectral width}

In this section we prove the inequality used in the main text that connects our chosen resource bound to the curvature of the space curve. For Hermitian $H$ and any operator $\mathcal O$:
\begin{align}
\|[H,\mathcal O]\| \le w(H)\,\|\mathcal O\|,
\label{eq:SI_comm_bound}
\end{align}
where $\|\cdot\|$ denotes the (normalized) Hilbert--Schmidt norm $\|A\|=\sqrt{\mathrm{Tr}(A^\dagger A)}$.

\begin{proof}
Begin by diagonalizing $H=\sum_a E_a\ket{a}\bra{a}$. Next express the matrix element of $[H,\mathcal{O}]$ in the eigenbasis of $H$. This reads $([H,\mathcal O])_{ab}=(E_a-E_b)\mathcal O_{ab}$ from which it follows that:
\begin{align}
\|[H,\mathcal O]\|^2
&=\frac{1}{n}\sum_{a,b}|E_a-E_b|^2\,|\mathcal O_{ab}|^2
\le \Big(\max_{a,b}|E_a-E_b|\Big)^2 \frac{1}{n}\sum_{a,b}|\mathcal O_{ab}|^2 \nonumber\\
&= w(H)^2\,\|\mathcal O\|^2.
\end{align}
The desired expression is then obtained by taking the square root of both sides.
\end{proof}

This result then connects to SCQC by recognizing $\kappa_{\mathcal O}(t)=\|[H(t),\mathcal O]\|$ for $\|\mathcal O\|=1$ which gives the bound quoted in the main text: 
\begin{align}
\kappa_{\mathcal O}(t)\le w(H(t))\le \Omega_{\max}.
\label{eq:SI_kappa_bound}
\end{align}

\subsection{Why the spectral width of $H$ provides tighter QSLs than $\|H\|$}
\label{Subsupp:1.3 Why the spectral width of H provides tighter QSLs that |H|}

The spectral width $w(H)$ provides a better measure of the rate of quantum dynamics compared to, e.g., the Hilbert-Schmidt norm $\|H\|$. One issue with the latter is that it is not invariant under nonphysical overall energy shifts $H(t)\to H(t)+\lambda(t)I$. Such shifts lead to irrelevant global phases in evolved states, or equivalently, in the evolution operator, and so have no physical impact on the quantum dynamics. One could avoid this issue by instead considering the Hilbert-Schmidt norm of the \emph{centered} Hamiltonian $\widetilde{H}(t)$, which is defined to be the same as $H(t)$ but with the overall energy shifted such that the instantaneous spectrum is centered around zero:
\begin{equation}
\widetilde H(t)=H(t)-\frac{E_{\max}(H(t))+E_{\min}(H(t))}{2}I,
\end{equation}
However, $\|\widetilde{H}(t)\|$ still fails to accurately capture the rate of quantum dynamics. To see this, consider the following example: $H=\widetilde{H}=\epsilon Z\otimes\mathbbm{1}_{n-1}$, where $\mathbbm{1}_{n-1}$ is the $(n-1)\times(n-1)$ identity operator. We have
\[
w(H)=2\epsilon,\qquad \|H\|=\sqrt{n}\epsilon.
\]
Even though the dynamics is only occurring within a single two-level subsystem, $\|H\|$ depends on the size of the entire Hilbert space and thus does not provide a faithful measure of the rate of dynamics. In contrast, $w(H)$ does not depend on the full Hilbert space dimension.

To see why $w(H)$ is the correct measure of the rate of quantum dynamics, we will show that the time derivative of any matrix element of $\widetilde U(t)$ (the evolution operator generated by the centered Hamiltonian $\widetilde{H}(t)$) is bounded above by $w(H)$. Consider the Schr\"odinger equation for $\widetilde{U}(t)$:
\begin{equation}
\frac{d}{dt}\widetilde U(t) = -i\,\widetilde H(t)\widetilde U(t),
\qquad
\widetilde U(0)=I.
\end{equation}
Fix basis vectors $\ket{a}$ and $\ket{b}$. Define the matrix element
\begin{equation}
\widetilde U_{ab}(t)=\mel{a}{\widetilde U(t)}{b}.
\end{equation}
Differentiating and using the Schr\"odinger equation for $\widetilde U(t)$ gives
\begin{equation}
\frac{d}{dt}\widetilde U_{ab}(t)
=
\mel{a}{\frac{d}{dt}\widetilde U(t)}{b}
=
-i\,\mel{a}{\widetilde H(t)\widetilde U(t)}{b}.
\end{equation}
Therefore,
\begin{equation}
\left|\frac{d}{dt}\widetilde U_{ab}(t)\right|
=
\left|\mel{a}{\widetilde H(t)\widetilde U(t)}{b}\right|.
\end{equation}
Apply the Cauchy--Schwarz inequality:
\begin{equation}
\left|\mel{a}{\widetilde H(t)\widetilde U(t)}{b}\right|
\le
\|\widetilde H(t)\widetilde U(t)\ket{b}\|.
\end{equation}
Since $\widetilde U(t)$ is unitary, $\|\widetilde U(t)\ket{b}\|=1$. Hence
\begin{equation}
\|\widetilde H(t)\widetilde U(t)\ket{b}\|
\le
\|\widetilde H(t)\|_\mathrm{op},
\end{equation}
where $\|\cdot\|_\mathrm{op}$ denotes the operator norm.
Combining the last two inequalities gives
\begin{equation}
\left|\frac{d}{dt}\widetilde U_{ab}(t)\right|
\le
\|\widetilde H(t)\|_\mathrm{op}.
\end{equation}
Since $\|\widetilde H(t)\|_\mathrm{op}=w(H)/2$ for a centered Hamiltonian, we have
\begin{equation}
\left|\frac{d}{dt}\mel{a}{\widetilde U(t)}{b}\right|
\le \frac{w(H(t))}{2}.
\end{equation}

\subsection{Gate certification and bottleneck observables}
\label{Subsupp:1.4 Gate certification and bottleneck observables}
A set $\mathcal G\subset \mathfrak{su}(n)$ is called \emph{certifying} if knowledge of $\text{Ad}_U(\mathcal O)=U^\dagger \mathcal O U$ for all $\mathcal O\in\mathcal G$ determines $U$ uniquely up to global phase.
For example, $\mathcal G$ may be a full operator basis of $\mathfrak{su}(n)$. For a fixed target $G$, any implementing protocol must satisfy the following endpoint constraints:
\begin{align}
U^\dagger(T)\mathcal O\,U(T)=G^\dagger \mathcal O\,G,\qquad \forall\,\mathcal O\in\mathcal G.
\label{eq:SI_cert_constraints}
\end{align}
Fix a certifying set $\mathcal G$.
Let $T_{\min}(\mathcal O\to \text{Ad}_G(\mathcal O))$ be the minimal time required under $w(H(t))\le \Omega_{\max}$
to satisfy the single constraint $U^\dagger(T)\mathcal O U(T)=\text{Ad}_G(\mathcal O)$.
Then any protocol implementing $G$ must satisfy:
\begin{align}
T \ge \max_{\mathcal O\in \mathcal G} T_{\min}(\mathcal O\to \text{Ad}_G(\mathcal O)).
\label{eq:SI_max_over_O}
\end{align}
This follows as all constraints \eqref{eq:SI_cert_constraints} must be satisfied simultaneously at the same final time $T$.
In particular, $T$ must be at least as large as the minimal time required for each individual constraint;
taking the maximum yields \eqref{eq:SI_max_over_O}.

\subsection{Two operators suffice to determine a unitary up to a global phase}
\label{Subsupp:1.5 Two operators suffice to determine a unitary up to a global phase}

Let $U(t) \in SU(n)$ and let
\[
\mathcal{O}_k(t)=U(t)^\dagger \mathcal{O}_k U(t) ,
\]
for fixed operators $\mathcal{O}_k \in \mathfrak{su}(n)$. We show that for every $n\ge 2$, there exist two operators $\mathcal{O}_1,\mathcal{O}_2$ such that the pair
\[
U(t)^\dagger \mathcal{O}_1 U(t),\; U(t)^\dagger \mathcal{O}_2 U(t)
\]
determines $U(t)$ uniquely up to a global phase.

\begin{theorem}
For every $n\ge 2$, there exist $\mathcal{O}_1,\mathcal{O}_2 \in \mathfrak{su}(n)$ such that if $U(t),V(t) \in SU(n)$ satisfy
\[
U(t)^\dagger \mathcal{O}_j U(t) = V(t)^\dagger \mathcal{O}_j V(t), \qquad j=1,2,
\]
then
\[
V(t) = \omega U(t)
\]
for some $\omega \in \mathbb{C}$ with $\omega^n=1$. In particular, $U(t)$ is determined up to a global phase. Moreover, one operator never suffices.
\end{theorem}

\begin{proof}
Work in the standard basis of $\mathbb{C}^n$. Define
\[
\mathcal{O}_1=\operatorname{diag}(1,2,\dots,n)-\frac{n+1}{2}I,
\]
and
\[
\mathcal{O}_2=\sum_{k=1}^{n-1}(E_{k,k+1}+E_{k+1,k}),
\]
where $E_{ab}$ denotes the matrix with a $1$ in entry $(a,b)$ and zeros elsewhere. Both $\mathcal{O}_1$ and $\mathcal{O}_2$ are Hermitian and traceless.

We first compute their common commutant. Let $A$ be a unitary matrix such that
\[
[A,\mathcal{O}_1]=[A,\mathcal{O}_2]=0.
\]
Since $\mathcal{O}_1$ has distinct eigenvalues, any matrix commuting with $\mathcal{O}_1$ must be diagonal in this basis:
\[
A=\operatorname{diag}(\lambda_1,\dots,\lambda_n).
\]
Now impose $[A,\mathcal{O}_2]=0$. For each $k=1,\dots,n-1$,
\[
AE_{k,k+1}=\lambda_k E_{k,k+1}, \qquad
E_{k,k+1}A=\lambda_{k+1}E_{k,k+1},
\]
so
\[
[A,E_{k,k+1}] = (\lambda_k-\lambda_{k+1})E_{k,k+1}.
\]
Thus $[A,\mathcal{O}_2]=0$ implies
\[
\lambda_k=\lambda_{k+1} \qquad \text{for all } k,
\]
hence
\[
\lambda_1=\cdots=\lambda_n=:\lambda.
\]
Therefore $A=\lambda I$. So the common commutant of $\mathcal{O}_1$ and $\mathcal{O}_2$ consists only of scalar matrices.

Now suppose
\[
U(t)^\dagger \mathcal{O}_j U(t) = V(t)^\dagger \mathcal{O}_j V(t),\qquad j=1,2.
\]
Set
\[
W(t)=V(t)U(t)^\dagger.
\]
Then
\[
W(t)^\dagger \mathcal{O}_j W(t) = \mathcal{O}_j,\qquad j=1,2,
\]
so $W(t)$ commutes with both $\mathcal{O}_1$ and $\mathcal{O}_2$. By the previous paragraph,
\[
W(t)=\lambda I
\]
for some $|\lambda|=1$. Therefore
\[
V(t)=\lambda U(t).
\]
Since $U(t),V(t) \in SU(n)$,
\[
1=\det(V(t))=\det(\lambda U(t))=\lambda^n \det(U(t))=\lambda^n,
\]
so $\lambda^n=1$. Hence $V(t)=\omega U(t)$ with $\omega$ in the center of $SU(n)$, proving uniqueness up to a global phase.

Finally, one operator never suffices for $n>1$. Indeed, the commutant of a single matrix is always larger than the scalars: if the matrix has a simple spectrum, every diagonal matrix in its eigenbasis commutes with it; if it has degeneracies, the commutant is even larger. Thus two operators are minimal.
\end{proof}

\subsection{Certifying set for deriving the quantum speed limit}
\label{Subsupp:1.6 Certifying set for deriving the quantum speed limit}

Let $\{|k\rangle\}_{k=1}^n$ be an orthonormal basis of $\mathbb{C}^n$. For $j\neq k$, define
\[
\mathcal{P}_{jk}
:=\frac{1}{\sqrt{2}}\bigl(|j\rangle\langle k|+|k\rangle\langle j|\bigr),
\qquad
\mathcal{Q}_{jk}
:=-\frac{i}{\sqrt{2}}\bigl(|j\rangle\langle k|-|k\rangle\langle j|\bigr).
\]
We show that the time-evolved operators
\[
U(t)^\dagger \mathcal{P}_{jk} U(t),
\qquad
U(t)^\dagger \mathcal{Q}_{jk} U(t),
\qquad j\neq k,
\]
determine $U(t)$ uniquely up to a global phase. The proof is similar to that in the previous section.

\begin{theorem}
Let $U,V\in U(n)$. Suppose that for all $j\neq k$,
\[
U^\dagger \mathcal{P}_{jk} U = V^\dagger \mathcal{P}_{jk} V,
\qquad
U^\dagger \mathcal{Q}_{jk} U = V^\dagger \mathcal{Q}_{jk} V.
\]
Then there exists $\lambda\in \mathbb{C}$ with $|\lambda|=1$ such that
\[
V=\lambda U.
\]
In particular, the family $\{\mathcal{P}_{jk},\mathcal{Q}_{jk}\}_{j\neq k}$ determines $U$ up to a global phase. If moreover $U,V\in SU(n)$, then $\lambda^n=1$.
\end{theorem}

\begin{proof}
For $j\neq k$, let
\[
E_{jk}:=|j\rangle\langle k|.
\]
From the definitions,
\[
\mathcal{P}_{jk}+i\mathcal{Q}_{jk}
=
\frac{1}{\sqrt{2}}(E_{jk}+E_{kj})
+\frac{1}{\sqrt{2}}(E_{jk}-E_{kj})
=
\sqrt{2}\,E_{jk}.
\]
Similarly,
\[
\mathcal{P}_{jk}-i\mathcal{Q}_{jk}
=
\sqrt{2}\,E_{kj}.
\]
Hence the collection $\{\mathcal{P}_{jk},\mathcal{Q}_{jk}\}_{j\neq k}$ is equivalent to the collection of all off-diagonal matrix units $\{E_{jk}\}_{j\neq k}$.

Now assume
\[
U^\dagger \mathcal{P}_{jk} U = V^\dagger \mathcal{P}_{jk} V,
\qquad
U^\dagger \mathcal{Q}_{jk} U = V^\dagger \mathcal{Q}_{jk} V
\]
for all $j\neq k$. Setting
\[
W:=VU^\dagger,
\]
we obtain
\[
W^\dagger \mathcal{P}_{jk} W = \mathcal{P}_{jk},
\qquad
W^\dagger \mathcal{Q}_{jk} W = \mathcal{Q}_{jk},
\qquad j\neq k.
\]
Therefore $W$ commutes with every $\mathcal{P}_{jk}$ and every $\mathcal{Q}_{jk}$, and hence with every $E_{jk}$ for $j\neq k$.

We now show that this forces $W$ to be a scalar multiple of the identity. Write
\[
W=(w_{ab})_{a,b=1}^n.
\]
Since $WE_{jk}=E_{jk}W$ for all $j\neq k$, compare matrix entries.

First, $W$ must be diagonal. Indeed, fixing $j\neq k$ and looking at the $(a,k)$ entry gives
\[
(WE_{jk})_{ak}=w_{aj},
\qquad
(E_{jk}W)_{ak}=\delta_{aj}w_{kk}.
\]
For $a\neq j$ this implies
\[
w_{aj}=0.
\]
Since $j$ was arbitrary, all off-diagonal entries vanish, so
\[
W=\operatorname{diag}(\lambda_1,\dots,\lambda_n).
\]

Next, all diagonal entries must be equal. For diagonal $W$,
\[
WE_{jk}=\lambda_j E_{jk},
\qquad
E_{jk}W=\lambda_k E_{jk}.
\]
Since $WE_{jk}=E_{jk}W$ for all $j\neq k$, we get
\[
\lambda_j=\lambda_k
\qquad \text{for all } j\neq k.
\]
Thus
\[
\lambda_1=\cdots=\lambda_n=:\lambda,
\]
and hence
\[
W=\lambda I.
\]
Because $W$ is unitary, $|\lambda|=1$.

Finally,
\[
V=WU=\lambda U,
\]
so $V$ differs from $U$ only by a global phase.

If $U,V\in SU(n)$, then
\[
1=\det(V)=\det(\lambda U)=\lambda^n\det(U)=\lambda^n,
\]
so $\lambda^n=1$. Thus the ambiguity is exactly the center of $SU(n)$.
\end{proof}

For $n\ge3$, the operators $\{\mathcal{P}_{jk}\}$ form a certifying set by themselves, as we now show.
\begin{theorem}
Assume $n\ge 3$. Let $U,V\in U(n)$ satisfy
\[
U^\dagger \mathcal{P}_{jk} U = V^\dagger \mathcal{P}_{jk} V
\qquad\text{for all } j\neq k.
\]
Then there exists $\lambda\in \mathbb{C}$ with $|\lambda|=1$ such that
\[
V=\lambda U.
\]
In particular, the family $\{\mathcal{P}_{jk}\}_{j\neq k}$ determines $U$ up to a global phase. If moreover $U,V\in SU(n)$, then $\lambda^n=1$.
\end{theorem}

\begin{proof}
For $j\neq k$, let
\[
S_{jk}:=|j\rangle\langle k|+|k\rangle\langle j|,
\]
so that $\mathcal{P}_{jk}=S_{jk}/\sqrt{2}$. It is therefore enough to work with the family $\{S_{jk}\}_{j\neq k}$.

Assume
\[
U^\dagger S_{jk} U = V^\dagger S_{jk} V
\qquad\text{for all } j\neq k,
\]
and define
\[
W:=VU^\dagger.
\]
Then
\[
W^\dagger S_{jk}W=S_{jk}
\qquad\text{for all } j\neq k,
\]
so $W$ commutes with every $S_{jk}$:
\[
WS_{jk}=S_{jk}W
\qquad\text{for all } j\neq k.
\]
We show that this implies $W=\lambda I$.

Write $W=(w_{ab})_{a,b=1}^n$. Fix $j\neq b$. Since $n\ge 3$, we may choose $k$ distinct from both $j$ and $b$. Comparing the $(k,b)$ entry of $WS_{jk}=S_{jk}W$ gives
\[
(WS_{jk})_{kb}=0,
\qquad
(S_{jk}W)_{kb}=w_{jb},
\]
hence
\[
w_{jb}=0.
\]
Since $j\neq b$ were arbitrary, all off-diagonal entries of $W$ vanish. Thus $W$ is diagonal:
\[
W=\operatorname{diag}(\lambda_1,\dots,\lambda_n).
\]

Now for diagonal $W$,
\[
WS_{jk}=\lambda_j |j\rangle\langle k|+\lambda_k |k\rangle\langle j|,
\]
while
\[
S_{jk}W=\lambda_k |j\rangle\langle k|+\lambda_j |k\rangle\langle j|.
\]
Since $WS_{jk}=S_{jk}W$, we obtain
\[
\lambda_j=\lambda_k
\qquad\text{for all } j\neq k.
\]
Therefore
\[
\lambda_1=\cdots=\lambda_n=:\lambda,
\]
so
\[
W=\lambda I.
\]
Because $W$ is unitary, $|\lambda|=1$. Hence
\[
V=WU=\lambda U,
\]
which proves that $V$ differs from $U$ only by a global phase.

If $U,V\in SU(n)$, then
\[
1=\det(V)=\det(\lambda U)=\lambda^n \det(U)=\lambda^n,
\]
so $\lambda^n=1$.
\end{proof}

Note that the restriction $n\ge 3$ is essential. For $n=2$ there is only one operator,
\[
\mathcal{P}_{12}=\frac{1}{\sqrt{2}}\bigl(|1\rangle\langle 2|+|2\rangle\langle 1|\bigr),
\]
and its commutant in $U(2)$ is larger than the scalar matrices, so it does not determine $U$ up to phase.

\section{Closed-form Pauli certifier geometry}
\label{Supp:2 Closed-form Pauli certifier geometry}

In this section we discuss the space curve geometry of two and three qubits in detail under the diagonal width-saturating, time-optimal generator $H_{\star,G}$.  In this frame, the adjoint action
$\text{Ad}_{U(t)}(\cdot)=U^\dagger(t)(\cdot)U(t)$ decomposes into independent two-dimensional rotation planes on orthogonal subspaces of the operator space. Consequently, every tangent curve generated by a fixed certifier $\mathcal O$ decomposes into a sum of planar circular motions, and the associated base curve is obtained by term-by-term integration. We then produce the space curves corresponding to the generation of the gates considered in the main text. 

\subsection{Universal two-dimensional plane decomposition for tangent curves}
\label{Subsupp:2.1 Universal two-dimensional plane decomposition for tangent curves}

Fix a constant generator $H$ and $U(t)=e^{-itH}$.  In the diagonal frame of $H$, the adjoint action splits into orthogonal $2\times 2$ rotation blocks. Hence for any $\mathcal O$ one can write:
\begin{equation}
\bm{\mathcal{O}} \cdot \dot{\textbf{r}}
=
\mathcal O_\parallel
+\sum_{\kappa_p\in\mathsf S(\mathcal O;H)}
\Big[\cos(\kappa_p t)\,A_{\kappa_p}+\sin(\kappa_p t)\,B_{\kappa_p}\Big],
\label{eq:SI_Ot_template_unified}
\end{equation}
where $\mathcal O_\parallel$ commutes with $H$, $\mathsf S(\mathcal O;H)$ is the set of nonzero plane curvatures supported by
$\mathcal O$, and each $(A_{\kappa_p},B_{\kappa_p})$ is a fixed orthonormal pair spanning the corresponding two-dimensional plane.
Integrating gives the base curve:
\begin{equation}
\bm{\mathcal O}\cdot \mathbf r(t)
=
t\,\mathcal O_\parallel
+\sum_{\kappa_p\in\mathsf S(\mathcal O;H)}
\left[
\frac{\sin(\kappa_p t)}{\kappa_p}\,A_{\kappa_p}
+\frac{1-\cos(\kappa_p t)}{\kappa_p}\,B_{\kappa_p}
\right].
\label{eq:SI_Rt_template_unified}
\end{equation}
Under the spectral-width constraint $w(H)\le \Omega_{\max}$, every curvature satisfies $\kappa_p\le w(H)\le \Omega_{\max}$.
For instance if $\mathcal O=A$ lies entirely in a single plane of rate $\Omega_{\max}$, one has the circular arc tangent:
\begin{equation}
U^\dagger(t)\,A\,U(t)=\cos(\Omega_{\max} t)\,A+\sin(\Omega_{\max} t)\,B,
\label{eq:SI_arc_tangent}
\end{equation}
and the integrated circular-arc base curve is:
\begin{equation}
\bm{\mathcal O}\cdot \mathbf r(t)
=\frac{1}{\Omega_{\max}}
\Big[\sin(\Omega_{\max} t)\,A+\big(1-\cos(\Omega_{\max} t)\big)\,B\Big].
\label{eq:SI_arc_curve}
\end{equation}
At $T_{\star,G}=\Delta\phi_{\star,G}/\Omega_{\max}$, the endpoint constraint is:
\begin{equation}
\text{Ad}_G(A)=U^\dagger(T_{\star,G})\,A\,U(T_{\star,G})
=\cos(\Delta\phi_{\star,G})\,A+\sin(\Delta\phi_{\star,G})\,B,
\label{eq:SI_arc_endpoint}
\end{equation}
so $A$ is a bottleneck observable for this certifier choice.

\subsection{Two qubit constructions}
\label{Subsupp:2.2 Two qubit constructions}

Fix $G\in SU(4)$ and its width-saturating shortest-log constant optimal generator $H_{\star,G}$, and work in the diagonal frame of $H_{\star,G}$. Any traceless diagonal two-qubit Hamiltonian lies in $\mathrm{span}\{ZI,IZ,ZZ\}$, so we may take the diagonal representative:
\begin{equation}
H_{\star,G} \;=\;\frac{1}{2}\big(a_1\,ZI + a_2\,IZ + a_3\,ZZ\big),
\label{eq:SI_cartan_diag_2q}
\end{equation}
with real coefficients $(a_1,a_2,a_3)$.
Let $\mathcal P:=\{\sigma_\mu\otimes\sigma_\nu:\mu,\nu\in\{I,X,Y,Z\}\}\setminus\{II\}$ denote the two-qubit Pauli strings and define
\begin{equation}
\ell_{\max}^{\mathcal P}(G)\;:=\;\max_{\substack{P\in\mathcal P\\ [P,H_{\star,G}]\neq 0}} \,\ell\!\big(P;H_{\star,G}\big),
\label{eq:SI_lmaxP_def_2q}
\end{equation}
where $\ell(P;H_{\star,G})$ is the Frenet--Serret closure dimension of the SCQC curve generated by the certifier $P$ under $H_{\star,G}$.

Under \eqref{eq:SI_cartan_diag_2q}, the $12$ non-diagonal Pauli strings partition into three invariant $4$D blocks:
\begin{align}
B_{12}=\{XX,YY,XY,YX\},\qquad
B_{13}=\{XI,XZ,YI,YZ\},\qquad
B_{23}=\{IX,ZX,IY,ZY\}.
\end{align}
Each block $B_{jk}$ splits into two orthogonal planes with signed plane curvatures given by $\kappa_{jk}^{\pm}:=a_j\pm a_k$ and indices $1,2,3$ corresponding to $ZI,IZ$, and $ZZ$ respectively. For $B_{23}$ one convenient orthonormal choice is:
\begin{equation}
\mathcal O_{1\pm}=\frac{IY\pm ZY}{\sqrt{2}},
\qquad
\mathcal O_{2\pm}=\frac{IX\pm ZX}{\sqrt{2}},
\label{eq:SI_B23_plane_basis}
\end{equation}
for which the adjoint dynamics closes as two independent planar rotations
\begin{equation}
\frac{d}{dt}
\begin{pmatrix}
U^\dagger \mathcal O_{1\pm}U\\[2pt]
U^\dagger \mathcal O_{2\pm}U
\end{pmatrix}
=
\begin{pmatrix}
0 & \kappa_{23}^{\pm}\\
-\kappa_{23}^{\pm} & 0
\end{pmatrix}
\begin{pmatrix}
U^\dagger \mathcal O_{1\pm}U\\[2pt]
U^\dagger \mathcal O_{2\pm}U
\end{pmatrix},
\qquad U(t)=e^{-it H_{\star,G}}.
\label{eq:SI_block_rotation_example_2q}
\end{equation}
Analogous constructions hold for $B_{12}$ and $B_{13}$ by permuting $(a_1,a_2,a_3)$. Every Pauli string $P\in B_{jk}$ has simultaneous nonzero overlap with the $+$ and $-$ planes. For example, in $B_{23}$,
\begin{equation}
IX=\frac{\mathcal O_{2+}+\mathcal O_{2-}}{\sqrt2},\qquad
ZX=\frac{\mathcal O_{2+}-\mathcal O_{2-}}{\sqrt2},\qquad
IY=\frac{\mathcal O_{1+}+\mathcal O_{1-}}{\sqrt2},\qquad
ZY=\frac{\mathcal O_{1+}-\mathcal O_{1-}}{\sqrt2}.
\label{eq:SI_Pauli_overlaps_2q}
\end{equation}
Thus $\ell(P;H_{\star,G})$ is determined by whether both plane curvatures are nonzero and distinct.
For each pair $(j,k)\in\{(1,2),(1,3),(2,3)\}$ define:
\begin{equation}
\Lambda(a_j,a_k):=
\begin{cases}
4,& a_j a_k\neq 0\ \text{and}\ |a_j|\neq|a_k|,\\[2pt]
3,& a_j a_k\neq 0\ \text{and}\ |a_j|=|a_k|,\\[2pt]
2,& a_j a_k=0\ \text{and}\ (|a_j|+|a_k|)\neq 0.
\end{cases}
\label{eq:SI_Lambda_def_2q}
\end{equation}
Then for \eqref{eq:SI_cartan_diag_2q},
\begin{equation}
\ell_{\max}^{\mathcal P}(G)
=\max_{(j,k)\in\{(1,2),(1,3),(2,3)\}}
\Lambda(a_j,a_k)
\in\{2,3,4\}.
\label{eq:SI_lmaxP_closed_form_2q}
\end{equation}
Therefore, for Pauli certifiers, the space curve dimension is determined by the relationships between the coefficients $a_1,a_2,a_3$. We now look at specific examples.

For the gate $\text{U}_{\text{ZX}}$ a diagonal-frame width-saturating generator reads:
\begin{equation}
H_{\star,\text{U}_{\text{ZX}}}=\frac{\Omega}{2}\,ZZ,
\qquad
T_{\star,\text{U}_{\text{ZX}}}=\frac{\pi}{2\Omega_{\max}},
\qquad
\Delta\phi_{\star,\text{U}_{\text{ZX}}}=\frac{\pi}{2}.
\label{eq:SI_Hstar_UZX_2q}
\end{equation}
One choice of bottleneck plane is: $\{ A_{ZX},B_{ZX} \} =\{(YI+YZ)/\sqrt{2},(XZ+XI)/\sqrt{2}\}$. Then the bottleneck circular arc can be written:
\begin{align}
\bm{\mathcal O}\cdot \dot{\mathbf{r}}_{A_{ZX}}(t)&=\cos(\Omega t)\,A_{ZX}+\sin(\Omega t)\,B_{ZX}\\
\bm{\mathcal O}\cdot \mathbf r_{A_{ZX}}(t) &=\frac{1}{\Omega}\Big[\sin(\Omega t)\,A_{ZX}+\big(1-\cos(\Omega t)\big)\,B_{ZX}\Big].
\end{align}
Here, $\Omega=E_{\max}-E_{\min}=\Omega_{\max}$ therefore at $T_{\star,\text{U}_{\text{ZX}}}$ one has $\text{Ad}_G(A_{ZX})=B_{ZX}$. For this gate the block $B_{12}$ does not evolve and similar dynamics hold in block $B_{23}$.

The diagonal generator corresponding to CNOT is:
\begin{equation}
H_{\star, \text{CNOT}}=\frac{\Omega}{2}\,(-ZI-IZ+ZZ),
\qquad
T_{\star,\text{CNOT}}=\frac{\pi}{\Omega_{\max}},
\qquad
\Delta\phi_\star=\pi,
\label{eq:SI_Hstar_CNOT}
\end{equation}
which generates $\mathrm{CZ}$ at $T_{\star,\text{CNOT}}$ up to global phase. Choose
\begin{equation}
A_{\mathrm{CZ}}:=\mathcal O_{2-}=\frac{IX-ZX}{\sqrt2},
\qquad
B_{\mathrm{CZ}}:=-\mathcal O_{1-}=\frac{ZY-IY}{\sqrt2},
\label{eq:SI_AB_CZ_arc}
\end{equation}
so that the rotation direction is absorbed into $B_{\mathrm{CZ}}$ and the rate is $2\Omega$:
\begin{align}
\bm{\mathcal O}\cdot \dot{\mathbf{r}}_{A_{\mathrm{CZ}}}(t) &= \cos(2\Omega t)\,A_{\mathrm{CZ}}+\sin(2\Omega t)\,B_{\mathrm{CZ}}\\
\bm{\mathcal O}\cdot \mathbf r_{A_{\mathrm{CZ}}}(t)&=\frac{1}{2\Omega}\Big[\sin(2\Omega t)\,A_{\mathrm{CZ}}+\big(1-\cos(2\Omega t)\big)\,B_{\mathrm{CZ}}\Big].
\label{eq:SI_CZ_arc}
\end{align}
Here the width of the Hamiltonian is $w(H_{\star,\text{CNOT}})=2 \Omega$ therefore saturation requires $2\Omega=\Omega_{\max}$ and at $T_{\star,\text{CNOT}}$ this yields $\text{Ad}_{\mathrm{CZ}}(A_{\mathrm{CZ}})=-A_{\mathrm{CZ}}$, requiring a rotation angle of $\pi$. We can additionally depict this increase in speed limit by a circular arc of length $\pi$ with twice the radius which is depicted in Fig. \ref{fig:circ_arcs}.

If we instead choose a Pauli string certifier $\mathcal O=IX$ then we obtain a helix:
\begin{equation}
U^\dagger(t)\,IX\,U(t)
=\frac{1}{2}(IX+ZX)
+\frac{1}{2}\Big[\cos(2\Omega t)\,(IX-ZX)+\sin(2\Omega t)\,(IY-ZY)\Big],
\label{eq:SI_CZ_helix_tangent}
\end{equation}
and integrating gives the base curve
\begin{equation}
\bm{\mathcal O}\cdot \mathbf r(t)
=t\,\frac{IX+ZX}{2}
+\frac{1}{4\Omega}
\left[
\sin(2\Omega t)\,(IX-ZX)
+\big(1-\cos(2\Omega t)\big)\,(IY-ZY)
\right].
\label{eq:SI_CZ_helix_curve}
\end{equation}
which is the helix depicted in Fig. \ref{fig:helices}.

The last gate we consider is $\text{U}_{\text{4d}}$ which is generated by: 
\begin{equation}
H_{\star,\text{U}_{4d}}=\frac{\Omega}{2}(2ZI+IZ),
\qquad
T_{\star,\text{U}_{4\text{d}}}=\frac{3\pi}{2\Omega_{\max}},
\qquad
\Delta\phi_{\star,\text{U}_{\text{4d}}}=\frac{3\pi}{2}.
\label{eq:SI_Hstar_U4d_2q}
\end{equation}
The $B_{12}$ plane can be rewritten in the following basis:
\begin{equation}
A^{-}_{4\mathrm d}:=\frac{XX-YY}{\sqrt2},
\qquad
B^{+}_{4\mathrm d}:=-\frac{XY+YX}{\sqrt2}, \qquad A^{+}_{4\mathrm d}:=\frac{XX+YY}{\sqrt2},
\qquad
B^{-}_{4\mathrm d}:=-\frac{XY-YX}{\sqrt2}.
\label{eq:SI_AB_U4d_arc}
\end{equation}
The bottleneck arc is obtained by conjugating $A^{-}_{4d}$ by $e^{-i H_{\star,\text{U}_{4d}}t}$:
\begin{align}
\bm{\mathcal O}\cdot \dot{\mathbf{r}}_{A^{-}_{4\mathrm d}}(t)
&=\cos(3\Omega t)\,A^-_{4\mathrm d}+\sin(3 \Omega t)\,B^{+}_{4\mathrm d}\\
\bm{\mathcal O}\cdot \mathbf r_{A^-_{4\mathrm d}}(t) &=\frac{1}{3\Omega}\Big[\sin(3\Omega t)\,A^-_{4\mathrm d}+\big(1-\cos(3\Omega t)\big)\,B^+_{4\mathrm d}\Big].
\end{align}
Here the width of the Hamiltonian is $w(H_{\star,\text{U}_{4d}})=3 \Omega$ therefore saturation requires $3\Omega=\Omega_{\max}$, therefore $T_{\star,\text{U}_{\text{4d}}}$, $\Omega_{\max}T_{\star,\text{U}_{\text{4d}}}=3\pi/2$ so $\text{Ad}_G(A^{-}_{4\mathrm d})=-B^{+}_{4\mathrm d}$, and here we can depict this increase in speed limit by a circular arc of length $3\pi/2$ with 3 times the radius which is depicted in Fig. \ref{fig:circ_arcs}. 

If we instead choose the Pauli string certifier $\mathcal O=XX$, which overlaps with both the  $+$ and $-$ planes we obtain an explicit two-frequency four-dimensional motion.
Using $XX=(A^-_{4d}+A^{+}_{4d})/\sqrt2$ with $A^{+}_{4d}=(XX+YY)/\sqrt2$, one finds:
\begin{align}
U^\dagger(t)\,XX\,U(t)
&=\frac{1}{\sqrt2}\Big[\cos(3 \Omega t)\,A^{-}_{4d}+\sin(3 \Omega t)\,B^{+}_{4d} \Big]
+\frac{1}{\sqrt2}\Big[\cos(\Omega t)\,A^{+}_{4d}+\sin(\Omega t)\,B^{-}_{4d}\Big],
\label{eq:SI_U4d_two_plane_tangent}
\end{align}
and integrating gives the base curve as a sum of two circular arcs:
\begin{align}
\bm{\mathcal O}\cdot \mathbf r(t)
&=\frac{1}{\sqrt2}\left[
\frac{\sin(3 \Omega t)}{3 \Omega}\,A^{-}_{4d}+\frac{1-\cos(3\Omega t)}{3\Omega}\,B^{+}_{4d}
\right]
+\frac{1}{\sqrt2}\left[
\frac{\sin(\Omega t)}{\Omega}\,A^{+}_{4d}+\frac{1-\cos(\Omega t)}{\Omega}\,B^{-}_{4d}
\right].
\label{eq:SI_U4d_two_plane_curve}
\end{align}
which is a four-dimensional helix. 

\subsection{Three qubit constructions}
\label{Subsupp:2.3 Three qubit constructions}

For three qubits ($n=8$), the same analysis can be applied. We begin with the diagonal Hamiltonian which can be expanded as:
\begin{equation}
H=\frac{1}{2}\Big( a_1\,ZII+a_2\,IZI+a_3\,IIZ+a_4\,ZZI+a_5\,ZIZ+a_6\,IZZ+a_7\,ZZZ \Big),
\label{eq:SI_Hdiag_3q}
\end{equation}
in terms of the seven commuting diagonal Pauli strings. Under conjugation by $U(t)=e^{-itH}$, all $Z$-type strings are left invariant, while the remaining $56$ Pauli strings
partition into seven eight dimensional blocks labeled by which qubits carry an $X/Y$:
\begin{align*}
B_{\{1\}} &= \{X,Y\}\otimes\{I,Z\}\otimes\{I,Z\}\\
B_{\{2\}} &= \{I,Z\}\otimes\{X,Y\}\otimes\{I,Z\}\\
B_{\{3\}} &= \{I,Z\}\otimes\{I,Z\}\otimes\{X,Y\},\\
B_{\{1,2\}} &= \{X,Y\}\otimes\{X,Y\}\otimes\{I,Z\}\\
B_{\{1,3\}} &= \{X,Y\}\otimes\{I,Z\}\otimes\{X,Y\}\\
B_{\{2,3\}} &= \{I,Z\}\otimes\{X,Y\}\otimes\{X,Y\},\\
B_{\{1,2,3\}} &= \{X,Y\}\otimes\{X,Y\}\otimes\{X,Y\}.
\end{align*}
For instance:
\begin{equation}
B_{\{3\}}=\{IIX,ZIX,IZX,ZZX,\,IIY,ZIY,IZY,ZZY\},
\end{equation}
and only the $Z$-on-qubit-$3$ terms in $H$ contribute, namely $IIZ,ZIZ,IZZ,ZZZ$  corresponding to the coefficients $a_3,a_5,a_6,a_7$. We can define for $\mu,\nu\in\{\pm1\}$ the unit-norm combinations:
\begin{align}
\mathcal O^{\mu\nu}_{1} &= \frac{1}{2}\Big( IIY + \mu\,ZIY + \nu\,IZY + \mu\nu\,ZZY \Big),\\
\mathcal O^{\mu\nu}_{2} &= \frac{1}{2}\Big( IIX + \mu\,ZIX + \nu\,IZX + \mu\nu\,ZZX \Big),
\label{eq:SI_Omunu_def}
\end{align}
which satisfy four independent planar rotations
\begin{equation}
\frac{d}{dt}
\begin{pmatrix}
U^\dagger \mathcal O^{\mu\nu}_{1}U\\[2pt]
U^\dagger \mathcal O^{\mu\nu}_{2}U
\end{pmatrix}
=
\begin{pmatrix}
0 & \kappa_{\mu\nu}\\
-\kappa_{\mu\nu} & 0
\end{pmatrix}
\begin{pmatrix}
U^\dagger \mathcal O^{\mu\nu}_{1}U\\[2pt]
U^\dagger \mathcal O^{\mu\nu}_{2}U
\end{pmatrix},
\qquad
\kappa_{\mu\nu}=a_3+\mu\,a_5+\nu\,a_6+\mu\nu\,a_7.
\label{eq:SI_3q_B3_2plane}
\end{equation}
Analogous four-plane decompositions hold for the other eight-dimensional blocks.

Within any block $B_S$, there are four plane frequencies. A Pauli string in that block has simultaneous overlap with multiple planes
(generically all four), so the maximal Pauli-certifier closure dimension is twice the number of distinct nonzero plane frequencies which in this case is generically eight. For both the generators $K_{\mathrm{GHZ}}=(XXX+ZII)/\sqrt2$ and $K_{\mathrm{W}}=(XYZ+YZX+ZXY)/\sqrt3$  the diagonal frame of the optimal generator can be taken as:
\begin{equation}
H_{\star,\text{U}_{\text{GHZ/W}}}=\frac{\Omega}{2}\,ZII,
\qquad
T_{\star,\text{U}_{\text{GHZ/W}}}=\frac{\pi}{\Omega_{\max}},
\qquad
\Delta\phi_{\star,\text{U}_{\text{GHZ/W}}}=\pi.
\label{eq:SI_Hstar_GHZW_3q}
\end{equation}
A convenient choice of bottleneck plane is then $\mathrm{span}\{XII,YII\}$. Choose the Pauli-combination pair
\begin{equation}
A_{\mathrm{GHZ/W}}:=\frac{XII+YII}{\sqrt2},
\qquad
B_{\mathrm{GHZ/W}}:=\frac{YII-XII}{\sqrt2},
\label{eq:SI_AB_GHZW}
\end{equation}
which yields the bottleneck arc:
\begin{align}
\bm{\mathcal O}\cdot \dot{\mathbf{r}}_{A_{\mathrm{GHZ/W}}}(t) &=\cos(\Omega t)\,A_{\mathrm{GHZ/W}}+\sin(\Omega t)\,B_{\mathrm{GHZ/W}}\\
\bm{\mathcal O}\cdot \mathbf r_{A_{\mathrm{GHZ/W}}}(t)&=\frac{1}{\Omega_{\max}}\Big[\sin(\Omega t)\,A_{\mathrm{GHZ/W}}+\big(1-\cos(\Omega t)\big)\,B_{\mathrm{GHZ/W}}\Big],
\end{align}
Similar to the $\text{U}_{ZX}$ gate above, $\Omega=E_{\max}-E_{\min}$ therefore we can set $\Omega_{\max}=\Omega$ and at $T_{\star,\text{U}_{\text{GHZ/W}}}$ one has $\text{Ad}_G(A_{\mathrm{GHZ/W}})=-A_{\mathrm{GHZ/W}}$ which is a $\pi$ rotation in the plane.

For Toffoli we can work in the diagonal frame of $\mathrm{CCZ}=\mathrm{diag}(1,1,1,1,1,1,1,-1)$ (locally equivalent to Toffoli).
A width-saturating traceless diagonal optimal generator is:
\begin{equation}
H_{\star,\mathrm{CCZ}}
=\frac{\Omega}{2}\Big(-ZII-IZI-IIZ+ZZI+ZIZ+IZZ-ZZZ\Big),
\qquad
T_{\star,\text{CCZ}}=\frac{\pi}{\Omega_{\max}},
\qquad
\Delta\phi_{\star,\text{CCZ}}=\pi.
\label{eq:SI_Hstar_CCZ_3q}
\end{equation}

In block $B_{\{3\}}$, only the $(-,-)$ plane is active with $|\omega_{--}|=4\Omega$.
A planar bottleneck certifier is:
\begin{equation}
A_{\mathrm{CCZ}}:=\mathcal O^{--}_{2}
=\frac{1}{2}\big(IIX - ZIX - IZX + ZZX\big),
\qquad
B_{\mathrm{CCZ}}:=-\mathcal O^{--}_{1}
=\frac{1}{2}\big(-IIY + ZIY + IZY - ZZY\big),
\label{eq:SI_AB_CCZ}
\end{equation}
where the sign in $B_{\mathrm{CCZ}}$ is chosen so the rotation rate is $4\Omega$. Then
\begin{align}
\bm{\mathcal O}\cdot \dot{\mathbf{r}}_{A_{\mathrm{CCZ}}}(t)&=\cos(4\Omega t)\,A_{\mathrm{CCZ}}+\sin(4\Omega t)\,B_{\mathrm{CCZ}}\\
\bm{\mathcal O}\cdot \mathbf r_{A_{\mathrm{CCZ}}}(t)&=\frac{1}{4\Omega}\Big[\sin(4\Omega t)\,A_{\mathrm{CCZ}}+\big(1-\cos(4\Omega t)\big)\,B_{\mathrm{CCZ}}\Big],
\label{eq:SI_CCZ_arc}
\end{align}
and at $T_{\star,\text{CCZ}}$ one has $\text{Ad}_{\mathrm{CCZ}}(A_{\mathrm{CCZ}})=-A_{\mathrm{CCZ}}$.

If we instead choose the Pauli string $\mathcal O=IIX$ we obtain an offset+rotation helix because $IIX$ overlaps with both inactive and active planes:
\begin{equation}
IIX=\frac{1}{4}\big(3\,IIX + ZIX + IZX - ZZX\big)
+\frac{1}{4}\big(IIX - ZIX - IZX + ZZX\big)
\label{eq:SI_IIX_decomp_CCZ}
\end{equation}
Thus the tangent takes the helix form
\begin{equation}
U^\dagger(t)\,IIX\,U(t)
=\mathcal O_\parallel+\frac{1}{2}(\cos(4\Omega t)\,A_{\text{CCZ}}+\sin(4\Omega t)\,B_{\text{CCZ}}),
\label{eq:SI_CCZ_helix_tangent}
\end{equation}
and the integrated base curve is:
\begin{equation}
\bm{\mathcal O}\cdot \mathbf r(t)
=t\,\mathcal O_\parallel
+\frac{1}{8\Omega}
\Big[\sin(4 \Omega t)\,A_{\text{CCZ}}+\big(1-\cos(4\Omega t)\big)\,B_{\text{CCZ}}\Big].
\label{eq:SI_CCZ_helix_curve}
\end{equation}
Here $4 \Omega$ is the maximum width of the Hamiltonian and therefore upon saturation $4 \Omega=\Omega_{\max}$, the required rotation angle is: $\Omega_{\max}T_{\star,\text{CCZ}} = \pi$.

\end{document}